\documentclass[11pt]{article}
\usepackage[top=1in,bottom=1in,left=1in,right=1in]{geometry}
\usepackage{amssymb,amsmath,enumerate,float,bbm,bigints}
\usepackage[colorlinks,linkcolor=blue,anchorcolor=blue,citecolor=blue]{hyperref}
\usepackage[dvips]{graphicx}
\usepackage[dvips]{graphics}
\usepackage[dvips]{epsfig}
\usepackage{theorem}
\usepackage{authblk}
\usepackage{lipsum}
\usepackage[normalem]{ulem} 
\usepackage[dvipsnames]{xcolor}

\newtheorem{theorem}{Theorem}[section]
\newtheorem{corollary}[theorem]{Corollary}
\newtheorem{proposition}[theorem]{Proposition}
\theorembodyfont{\upshape}
\newtheorem{definition}[theorem]{Definition}

\theorembodyfont{\upshape}

\newtheorem{remark}[theorem]{Remark}

\def\bproposition{\begin{proposition}}\def\eproposition{\end{proposition}}
\def\beqlb{\begin{eqnarray}}\def\eeqlb{\end{eqnarray}}
\def\beqnn{\begin{eqnarray*}}\def\eeqnn{\end{eqnarray*}}

 \def\qed{\hfill$\Box$\medskip}

\begin{document}

\title{Closed-End Formula for options linked to Target Volatility Strategies} 
\author[a]{Luca DI PERSIO}
\author[b,c]{Luca PREZIOSO}
\author[d]{Kai WALLBAUM}
\affil[a]{Department of Computer Science, University of Verona, Strada le Grazie, 15, Verona, Italy}
\affil[b]{Department of Mathematics, University of Trento, via Sommarive, 14, Trento, Italy}
\affil[c]{LPSM, University of Paris Diderot, 5 Rue Thomas Mann, Paris, France}
\affil[d]{RiskLab, Allianz Global Investors, Seidlstrasse 24-24a, Munchen, Germany}

\maketitle

\begin{description}
\item[Abstract.] Recent years have seen an emerging class of structured financial products based on options linked to dynamic asset allocation strategies. One of the most chosen approach is the so-called target volatility mechanism. It shifts between risky and riskless assets to control the volatility of the overall portfolio. Even if a series of articles have been already devoted to the analysis of options linked to the target volatility mechanism, this paper is the first, to the best of our knowledge, that tries to develop closed-end formulas for VolTarget options. In particular, we develop closed-end formulas for option prices and some key hedging parameters within a Black and Scholes setting, assuming the underlying follows a target volatility mechanism.
\item[Key words.] Volatility target portfolio, generalized Black-Scholes model, European options, exact formulas, Greeks, Euler-Maruyama scheme, Milstein scheme.
\item[AMS subject classification.] 91G10, 91G20, 91G50.
\end{description}

\section{Introduction}

In the aftermath of the financial markets, risk management solutions became more and more important for institutional and retail investors. The low interest rate environment forced practitioners to think of more efficient techniques how to use the available risk budgets in client portfolios. One of the most successful strategies, which were introduced in multi asset portfolios but also within structured products offering at least partial capital protection is the so-called target volatility strategy, also known as {\it VolTarget strategy} (VTS). This concept shifts dynamically between risk-free and risky assets in order to generate a portfolio with a stable risk level independent of market volatilities. The approach assumes that market volatilities are a good indicator for asset allocation decisions and the concept works well in rising markets with low volatility and in falling markets with higher volatilities. Practitioner often compare the concept with a constant portfolio protection insurance (CPPI) strategy, which also allocated dynamically between risky and riskless assets, but in this concept the investment process aims to achieve capital protection in general.
 
In recent years, dynamic asset allocation process like VTS or CPPI Strategies have been used as underlying of options and we saw a series of academic papers looking into option theory when the underlying of the derivative follows a certain trading rule shifting between risky and riskless assets. Most of these papers took a numeric approach to determine option prices or hedging parameters. We refer, e.g. to Albeverio et al. \cite{ASW13,ASW17}, Jawaid \cite{Jawaid15,Jawaid16}, Zakamulin \cite{Zakamulin14}, who were especially looking at VTS in different market models. Zagst et al. \cite{EKSZ11} focused on option on a CPPI and they also developed a closed-end formula of CPPI options in a Black-Scholes environment. 
 
This paper is the first attempt, to the best of ur knowledge, which considers closed-end formulas for VTS-linked options. 
Our underlying environment can be compared to the one by Zagst et al. in \cite{EKSZ11}, where the authors assume the risky asset to evolve as a Black-Scholes model.
We extend such an analysis  considering a generalized geometric Brownian motion framework, with random drift and diffusion, adapted to the real-world filtration of the probability space, then deriving  a closed-end formula for call and put options linked to VTS portfolios. We also consider the model's {\it Greeks}, providing closed-end expressions for key hedging parameters of options linked to VTSs. We would llike to stress that our results constitute an important step for any practitioner, who is pricing and hedging options linked to VTS portfolios.
 
The paper is organized as follows: in section \ref{CoV:1} we analyzed VTS portfolios in the case that the risky-asset dynamics are described by a generalized geometric Brownian motion. 
We treated the evaluation problem for options that have, as underlying, VTS portfolios determined by standard VTS, preserving fixed volatility in time. In section \ref{s:lev} we considered a  modification of the VTS, which is placing an upper bound to the leverage effect caused by the continuous dynamic adjustments of the VTS. In this case we considered a risky asset evolving as a geometric Brownian motion with time-dependent drift and volatility. For both strategies exact formulas for the price of call and put options are presented.
In section \ref{greeks} we analyzed the sensibility of the prices of options written on VTS portfolios with respect to volatility and risky asset value. We gave emphasis on the analysis of the Greeks (Vega, Delta and Gamma) with respect to changes in the underlying volatility, also providing several graphs to better highlight both the robustness and soundness  of our results.
In section \ref{numerics} we present relevant simulations for the  paths of a VTS portfolio, by exploiting both Euler-Maruyama and Milstein discretization approach, when the risky asset is described by a Heston model.


\section{Standard VolTarget strategy in a generalized Geometric Brownian motion environment}\label{CoV:1}
Let us start by considering a framework similar to the pioneering paper by Merton \cite{Merton71}. This means that through the paper we are going to consider a market in which two investment opportunities occurs: a {\it risk-free asset}, also referred as money market or Government bond or simply bond, and a {\it risky underlying asset}, also called stock or share. Moreover we assume that the randomness of the underlying asset is described by Black-Scholes-Merton stochastic differential equations and that there exist continuously-trading perfect markets where the agents are not subjected to any transaction costs to trade the risky asset for the riskless asset and vice versa.

\bigskip

Let $(\Omega,\mathcal{F},\{\mathcal{F}(t)\}_{t\ge0},\mathbb{P})$ be a filtered, complete probability space, with right-continuous filtration,
supporting a Brownian motion $W$, and consider a market consisting of two investment opportunities: a risky asset $\{S(t)\}_{t\ge0}$, and a riskless asset $\{B(t)\}_{t\ge0}$, evolving as a stochastic process satisfying the generalized geometric Brownian motion and a deterministic function:
\begin{align}
\mathrm{d}S(t)&=S(t)\,\bigl(\mu(t)\,\mathrm{d}t+\sigma(t)\,\mathrm{d}W(t)\bigr),\label{risky-asset}\\
\mathrm{d}B(t)&=r\,B(t)\,\mathrm{d}t,\nonumber
\end{align}
for each $t\ge t_0$, where $t_0\ge0$ is the starting time, $W$ is a Brownian motion $\mathcal{F}(t)$-adapted, $r\in\mathbb{R}^+$ is a positive constant representing the risk-free rate, and $\mu$ and $\sigma$ are stochastic processes adapted to $\{\mathcal{F}^W\}$, the natural filtration generated by the Brownian motion, and represent the mean rate of return and the percentage volatility of the risky asset respectively. Let $s,b\in\mathbb{R}_+$ be the values at time $t_0$ for the risky and riskless assets.

Moreover, consider an investor holding a portfolio, starting with a positive position $x$ invested in the riskless asset and a positive position $y$ invested in the risky asset, and assume that he is able to 
 transfer its capitals from an investment to another without paying any transaction costs. Therefore let us introduce the processes $L$ representing the cumulative amount of riskless asset sold in order to buy risky asset, and $M$ the process representing the cumulative amount of risky asset sold in order to buy riskless asset. Both $L$ and $M$ are assumed to be non-negative, non-decreasing and c\`adl\`ag.\\
Finally, the portfolio value can be represented continuously in time by the couple $(X(t),Y(t))_{t\ge t_0}$, starting at $X(t_0)=x$, $Y(t_0)=y$, with $(X,Y)$ representing the amount of capital invested in the riskless asset and in the risky asset, respectively, and evolving according to the following stochastic differential equations 
\begin{align}
\mathrm{d}X(t)&=r\,X(t)\,\mathrm{d}t+\mathrm{d}M(t)-\mathrm{d}L(t),\label{X}\\
\mathrm{d}Y(t)&=\mu(t)\,Y(t)\,\mathrm{d}t+\sigma(t)\,Y(t)\,\mathrm{d}W(t)+\mathrm{d}L(t)-\mathrm{d}M(t).\label{Y}
\end{align}
In what follows we are not introducing neither proportional or fixed transaction costs, leaving these framework for a future study. Let us just outline that, if costs are taken into account, then continuous re-balancing would cause non-negligible expenses to the investor in the portfolio aiming to preserve a fixed volatility which implies to consider a volatility target interval, instead of a punctual volatility target.

let us denote the total portfolio value of the investor at time $t>t_0$, by $V(t)$, and let $\alpha(t)$ denotes the percentage of  portfolio invested at the same time in the risky-asset assumed to be an adapted predictable c\`adl\`ag processe, while $1-\alpha(t)$ will denote the portfolio weight invested in the riskless-asset, namely
we define $V(t)=X(t)+Y(t)$, while $\alpha(t)=\frac{Y(t)}{X(t)+Y(t)}$.
Since the investments evolve according to \eqref{X} and \eqref{Y}, by substituting the risky asset dynamics \eqref{risky-asset}, we derive the dynamics of the portfolio value process:
\begin{equation}\label{wealth}
\begin{cases}
\mathrm{d}V(t)=V(t)\bigl\{\bigl(\alpha(t)\,(\mu(t)-r)+r\bigr)\,\mathrm{d}t+\alpha(t)\,\sigma(t)\,\mathrm{d}W(t)\bigr\},\quad t>t_0,\\
V(0)=x+y=:v,
\end{cases}
\end{equation}
where $\alpha$ is controlled by the investor and adapted to the filtration $\mathcal{F}$. We make two remarks: first of all notice that the portfolio whose value is determined by \eqref{wealth} is self-financing, i.e. the dynamics of \eqref{wealth} are equivalent to
\begin{equation}\label{sf:VTp}
\mathrm{d}V(t)=V(t)\left(\alpha(t)\,\frac{\mathrm{d}S(t)}{S(t)}+(1-\alpha(t))\,\frac{\mathrm{d}B(t)}{B(t)}\right),
\end{equation}
moreover notice that $V$ is a Markovian portfolio and that a priori we do not know the future value of the wealth process since it has a random dynamic.

Turning back to the Volatility Target (VT) investment strategy, let us recall that it is a dynamic asset allocation mechanism, where the amount invested in the risky asset is determined by a pre-defined volatility target level, denoted by $\widehat{\sigma}$ representing the volatility of the underlying risky asset, $\sigma(t)$, see \eqref{risky-asset}. By  dynamically shifting between the two investment opportunities, which evolve accordingly to equations \eqref{risky-asset}, the investor aims at preserving a constant volatility level of the resulting portfolio $V_{\widehat{\sigma}}(t)$, which can be used as underlying for derivatives, e.g., for European call/put options. We resume this notion in the following definition:

\begin{definition}[VolTarget strategy portfolio]\label{D:VT}
Consider the stochastic process $V_{\widehat{\sigma}}$ evolving according to
\begin{equation}\label{Valpha}
V_{\widehat{\sigma}}(t)=X_{1-\widehat{\alpha}}(t)+Y_{\widehat{\alpha}}(t),\quad t\ge t_0,
\end{equation}
where $Y_{\widehat{\alpha}}=\widehat{\alpha}\,V_{\widehat{\sigma}}$ and $X_{\widehat{\alpha}}=(1-\widehat{\alpha})\,V_{\widehat{\sigma}}$, and by $\widehat{\alpha}$ we meant the proportion of portfolio value $\alpha(t)$ dynamically invested in the risky asset. 
We say that $V_{\widehat{\sigma}}$ is a {\it VolTarget strategy portfolio} if it is self-financing and the weight process is  preserving a constant volatility equal to $\widehat{\sigma}$, where $X$ and $Y$ represent the amount of capital invested in the riskless and risky asset and evolve according to \eqref{X} and \eqref{Y}, respectively. 
\end{definition}

We want to determine explicitly the equation for the control which preserves a fixed volatility to the portfolio process \eqref{wealth}.
\begin{proposition}\label{Prop:VT}
For $\widehat{\alpha}(t)=\widehat{\sigma}/\sigma(t)$ we have that the process whose dynamics are given by \eqref{Valpha} is a VTS portfolio.
\end{proposition}
\begin{proof}
By \eqref{wealth}, for $\alpha(t)=\widehat{\alpha}(t)$, we have 
\begin{equation}\label{V:VT}
\mathrm{d}V_{\widehat{\sigma}}(t)=V_{\widehat{\sigma}}(t)\left(\left(\frac{\widehat{\sigma}}{\sigma(t)}\,(\mu(t)-r)+r\right)\,\mathrm{d}t+\widehat{\sigma}\,\mathrm{d}W(t)\right),
\end{equation}
i.e. by Definition \ref{D:VT} we reach our thesis.
\qed
\end{proof}

Notice that we are not considering the equation involving the underlying \eqref{risky-asset} and the amounts of capital that have to be invested in the risky and riskless asset respectively. 
In  Proposition \ref{Prop:VT} we saw that in order to obtain a VTS portfolio the investor has to keep this ratio inversely proportional to the actual value of the volatility rate of the risky asset, i.e. equal to $\widehat{\alpha}(t)=\widehat{\sigma}/\sigma(t)$, which is stochastic.

\subsection{Option pricing}
Let $\mathcal{X}=\Phi(V_T)$ be a {\it contingent claim} with maturity $T$ and with underlying portfolio $V$, where $\Phi$ is a {\it contract function}. 
In what follows we provide an arbitrage-free price $\Pi(t;\mathcal{X})$ for such a claim, sometimes also denoted as $\Pi(t;\Phi)$ or $\Pi(t)$.

\bigskip

Let us heuristically assume for the moment that there exists a function $F\in C^{1,2}([0,T]\times\mathbb{R}_+)$ such that
\[
\Pi(t)=F(t,S(t)),
\]
then by the Black-Scholes equation we would have absence of arbitrage if $F$ is solution to the following PDE
\begin{equation}\label{BS:eq}
\begin{cases}
\partial(t)F(t,s)+r\,s\,\partial(s)F(t,s)+\frac{1}{2}\,s^2\,\widehat{\sigma}^2\,\partial_{ss}^2F(t,s)-r\,F(t,s)&=0\\
F(T,s)&=\Phi(s),
\end{cases}
\end{equation}
for $t\in[0,T]$ and $s\in\mathbb{R_+}$.
In the next subsections we are going to remove the above mentioned heuristic assumption to derive  pricing formulas for contingent claims written on the VTS portfolio. Let us note that the associated PDE \eqref{BS:eq}, can be solved \`a la Feynman-Ka\v{c}, i.e.
\[
F(t,s)=e^{-r\,(T-t)}\,\mathbb{E}_{t,s}^{\mathbb{Q}}[N(S_T)],\quad\text{for }t\in[0,T], s>0,
\]
where we are considering the expectation with respect to the unique risk-neutral measure $\mathbb{Q}$ conditioned by $S(t)=s$, see, e.g., \cite[Ch. 14]{bjork} for what concerns the existence of a unique risk-neutral measure.

\subsubsection{Risk Neutral Valuation}
Let $\mathbb{Q}$
be the unique equivalent martingale measure 
, namely the unique measure under which $S(t)/B(t)$ is a local martingale, and let $W^\mathbb{Q}$ be a Brownian motion under $Q$. Then, by Girsanov theorem, we have that the risky asset process $S$ satisfies the following SDE
\begin{equation}\label{RA_Gir}
\mathrm{d}S(t)=S(t)\left(r\,\mathrm{d}t+\sigma(t)\,\mathrm{d}W^{\mathbb{Q}}(t)\right)\,.
\end{equation}
Since the underlying risky asset is governed by a geometric Brownian motion with dynamics given by equation \eqref{RA_Gir}, we can apply the It\^o-D\"oblin formula to $\log(S(t))$ obtaining 
\begin{equation*}
S(t)=S(0)\,\exp\left\{\int_{0}^t(r-\sigma(s)^2/2)\,\mathrm{d}s+\int_{0}^t\sigma(s)\,\mathrm{d}W^{\mathbb{Q}}(s)\right\}\,.
\end{equation*}
Since the volatility in equation \eqref{RA_Gir} is stochastic, we cannot say much about the distribution of $\log(S(t)/S(0))$, which would have been Gaussian, in the special case of deterministic volatility.

Let us consider a European call, resp. put,  option with payoff 
\begin{align}\label{payoff}
\Phi_{\text{call}}(V_T)=(V_T-K)_+,\\
\Phi_{\text{put}}(V_T)=(K-V_T)_+,\label{payoff:p}
\end{align}
$T\ge t_0$ being its maturity time, while  $K$ represents its strike price. Through the next propositions and corollaries we will determine the price at the starting time $t_0\ge0$ of these kinds of options within the VolTarget framework.

\begin{proposition}\label{CoV:Prop2}
Assuming that the risky asset dynamics follow a generalized geometric Brownian motion with random $\mathcal{F}^W$-adapted drift and volatility, see equation \eqref{risky-asset}, the price at time $t_0$ of a call option with payoff \eqref{payoff}, denoted as $\Phi_{\text{call}}$, linked to the VTS portfolio $V_{\widehat{\sigma}}(t)$, see equation \eqref{V:VT}, is given by the following explicit formula
\begin{equation}\label{call:VT}
\Pi(t_0,\Phi_{\text{call}}(V_{\widehat{\sigma}}(T)))=v\,N\left(d_1(t_0)\right)-K\,e^{-r\,(T-t_0)}\,N\left(d_2(t_0)\right),
\end{equation}
where we recall that the proportion $\widehat{\alpha}(t)$ of portfolio value invested in the risky-asset is as the one defined in Proposition \ref{Prop:VT}, $N$ is the cumulative distribution function for the standard normal distribution, $v=V_{\widehat{\sigma}}(t_0)$ is the starting value of  the portfolio and we defined the following parameters
\begin{align*}
&d_1(t_0)=\frac{-z_{\widehat{\sigma}}(t_0)+\widehat{\sigma}\,(T-t_0)}{\sqrt{T-t_0}},\\
&d_2(t_0)=-\frac{z_{\widehat{\sigma}}(t_0)}{\sqrt{T-t_0}},\\
&z_{\widehat{\sigma}}(t_0)=\frac{1}{\widehat{\sigma}}\log\left(\frac{K}{v}\right)+\left(\frac{\widehat{\sigma}}{2}-\frac{r}{\widehat{\sigma}}\right)(T-t_0).\\
\end{align*}
\end{proposition}

\begin{proof}
Notice that while the underlying risky asset has non constant volatility, see eq. \eqref{V:VT}, the dynamics for the VTS portfolio are simpler. In fact, we can easily obtain the following explicit solution
\[
V_{\widehat{\sigma}}(t)=v\,\exp\left(\left(r-\frac{\widehat{\sigma}^2}{2}\right)(t-t_0)+\widehat{\sigma}\,W^{\mathbb{Q}}(t-t_0)\right),\quad\text{for }t\ge t_0\,.
\]
Therefore, we have that $V_{\widehat{\sigma}}(T)>K$ iff 
\[
W^{\mathbb{Q}}(T-t_0)>\frac{1}{\widehat{\sigma}}\log\left(\frac{K}{v}\right)+\left(\frac{\widehat{\sigma}}{2}-\frac{r}{\widehat{\sigma}}\right)(T-t_0)=:z_{\widehat{\sigma}}(t_0)\,.
\]
Denoting by $f_{N(0,t)}(x)$ the probability density function of the Gaussian random variable havig mean $0$ and variance $t$, i.e. :
\[
f_{N(0,t)}(x)=\frac{1}{\sqrt{2\pi\,t}}\,e^{-\frac{x^2}{2t}}\,,
\]
and by $N(x)$ the cumulative distribution function of a standard Gaussian random variable, then we have
that the price of the call option on the portfolio value at time $t_0$ equals
\begin{align*}
\Pi(t_0,\Phi_{\text{call}}(V_{\widehat{\sigma}}(T)))&=\mathbb{E}\left[e^{-r\,(T-t_0)}(V_{\widehat{\sigma}}(T)-K)_+\Big|\mathcal{F}_{t_0}\right]\\
&=e^{-r\,(T-t_0)}\int_{z_{\widehat{\sigma}}(t_0)}^{+\infty}\left\{v\,\exp\left(\left(r-\frac{\widehat{\sigma}^2}{2}\right)(T-t_0)+\widehat{\sigma}\,x\right)-K\right\}\,f_{N(0,T-t_0)}(x)\,\mathrm{d}x\\
&=e^{-r\,(T-t_0)}\,v\,e^{\left(r-\frac{\widehat{\sigma}^2}{2}\right)(T-t_0)+\widehat{\sigma}^2/2\,(T-t_0)}\left(1-N\left(\frac{z_{\widehat{\sigma}}(t_0)-\widehat{\sigma}\,(T-t_0)}{\sqrt{T-t_0}}\right)\right)\\
&\qquad-K\,e^{-r\,(T-t_0)}\left(1-N\left(\frac{z_{\widehat{\sigma}}(t_0)}{\sqrt{T-t_0}}\right)\right)\\
&=v\,N\left(\frac{-z_{\widehat{\sigma}}(t_0)+\widehat{\sigma}\,(T-t_0)}{\sqrt{T-t_0}}\right)-K\,e^{-r\,(T-t_0)}\,N\left(-\frac{z_{\widehat{\sigma}}(t_0)}{\sqrt{T-t_0}}\right).
\end{align*}
\qed
\end{proof}

\begin{corollary}\label{CoV:Cor}
Assuming that the risky asset dynamics follow a generalized geometric Brownian motion with random $\mathcal{F}^W$-adapted drift and volatility, see equation \eqref{risky-asset}, the price at time $t_0$ of a put option with payoff \eqref{payoff:p}, denoted as $\Phi_{\text{put}}$, linked to the VTS portfolio $V_{\widehat{\sigma}}(t)$, see equation \eqref{V:VT}, is given by the following explicit formula
\begin{equation}
\Pi(t_0,\Phi_{\text{put}}(V_{\widehat{\sigma}}(T)))=K\,e^{-r\,(T-t_0)}\,N\left(-d_2(t_0)\right)-v\,N\left(-d_1(t_0)\right)\,,\label{put:VT}
\end{equation}
where the parameters $d_1$, $d_2$ and $z_{\widetilde{\sigma}}$ are defined as in Proposition \ref{CoV:Prop2}.
\end{corollary}

\begin{proof}
By the {\it put-call parity} formula, see, e.g., \cite[4.5.6]{Shreve}, we have that the difference between the price of a call option and the price of put option with same strike price, time to expiration and underlying, equals the difference between the actual price of the underlying, represented by the VT portfolio in our setting,  and the discounted strike price, namely
\[
\Pi(t_0,\Phi_{\text{call}}(V_{\widehat{\sigma}}(T)))-\Pi(t_0,\Phi_{\text{put}}(V_{\widehat{\sigma}}(T)))=v-K\,e^{-r\,(T-t_0)}\,, 
\]
therefore, we have \eqref{put:VT}, since $N(-x)=1-N(x)$ for each $x\in\mathbb{R}$.  
\qed
\end{proof}

\section{VolTarget Strategy with maximum allowed Leverage Factor}\label{s:lev}
In what follows we are going to study
 a more interesting strategy from the {\it practitioner} point of view. In particular, we introduce a parameter $L\ge 1$ determining the maximum allowed leverage of the portfolio, i.e. we force the weight process to be less or equal than the parameter $L$:
\begin{equation}\label{alpha:L:t}
\widetilde{\alpha}(t):=\min\{L;\widehat{\sigma}/\sigma(t)\}.
\end{equation}
From now on, we will distinguish the notations for standard VTSs 
by the one for VTSs with maximum allowed leverage factor (MLVTS, in short), by marking the volatility and weight symbols with an hat and a tilde, respectively.
In particular, while $\widehat{\sigma}$ and $\widehat{\alpha}$ refere to standard VTS portfolios, $\widetilde{\sigma}$ and $\widetilde{\alpha}$ are referred to MLVTS portfolios.

This limitation is imposed in order to prohibit VTSs that finance by loans a large
portion of the risky investment. The typical setup occurring
within real world scenarios is $L = 2$, see \cite{ASW17}, for further details.

The next proposition gives an analytical expression for the value of a European call option written on the MLVTS portfolio, with limited leverage, and time dependent volatility. In particular, we are considering a particular case of equation \eqref{risky-asset}: 
\begin{equation}\label{risky-asset-t}
\mathrm{d}S(t)=S(t)\bigl(\mu(t)\,\mathrm{d}t+\sigma(t)\,\mathrm{d}W(t)\bigr)\,,
\end{equation}
where $\mu,\sigma:\mathbb{R^+}\rightarrow\mathbb{R}^+$ are deterministic functions of time, allowing, instead, the percentage drift term to be  stochastic and $\mathcal{F}^W$-adapted.

\begin{proposition}\label{Prop:C:VTl}
Assuming that the risky asset dynamics follows a geometric Brownian motion with time-dependent drift and volatility, see equation \eqref{risky-asset-t}, the price at time $t_0$ of a call option with payoff \eqref{payoff}, denoted as $\Phi_{\text{call}}$, linked to the MLVTS portfolio $V_{\widetilde{\sigma}}(t)$, is given by the following explicit formula
\begin{equation}\label{call:VTl}
\Pi\left(t_0,\Phi_{\text{call}}(V_{\widetilde{\sigma}}(T))\right)=v\,N\left(\widetilde{d}_1(t_0)\right)-K\,e^{-r\,(T-t_0)}\,N\left(\widetilde{d}_2(t_0)\right)
\end{equation}
where the proportion of portfolio value invested in the risky-asset is $\widetilde{\alpha}(t):=\min\{L;\widehat{\sigma}/\sigma(t)\}$, $N$ is the cumulative distribution function for the standard normal distribution, $v=V_{\widehat{\sigma}}(t_0)$ is the starting value of  the portfolio and we defined the following parameters
\begin{align*}
&\widetilde{d}_1(t_0)=\frac{-\widetilde{z}_{\widehat{\sigma}}(t_0)+\varsigma(t_0)}{\sqrt{\varsigma(t_0)}},\\
&\widetilde{d}_2(t_0)=-\frac{\widetilde{z}_{\widehat{\sigma}}(t_0)}{\sqrt{\varsigma(t_0)}},\\
&z_{\widetilde{\sigma}}(t_0)=\log\left(\frac{K}{v}\right)-r\,(T-t_0)+\frac{\varsigma(t_0)}{2},\\
&\varsigma(t_0)=\int_{t_0}^T\widetilde{\sigma}(s)^2\,\mathrm{d}s,\\
&\widetilde{\sigma}(t)=\min\{L\,\sigma(t),\widehat{\sigma}\}.
\end{align*}
\end{proposition}

\begin{proof}
For this strategy we have that the portfolio value has not a constant volatility and it has the following expression
\[
V_{\widetilde{\sigma}}(t_0,t)=v\,\exp\left(r\,(t-t_0)-\varsigma(t_0)/2+\int_{t_0}^t\min(L\,\sigma(s),\widehat{\sigma})\,\mathrm{d}W^\mathbb{Q}(s)\right),
\]
and we have $\widetilde{W}(t-t_0):=\int_{t_0}^t\min(L\,\sigma(s),\widehat{\sigma})\,\mathrm{d}W^\mathbb{Q}(s)\sim N(0,\varsigma(t_0))$, which means that its probability density function is 
\[
f_{N(0,\varsigma(t_0))}(x)=\frac{1}{\sqrt{2\,\pi\,\varsigma(t_0)}}\exp\left(-\frac{x^2}{2\,\varsigma(t_0)}\right).
\]
Therefore we have that $V_{\widetilde{\sigma}}(T)>K$ iff
\[
\widetilde{W}(T-t_0)>\log\left(\frac{K}{v}\right)-r\,(T-t_0)+\frac{\varsigma(t_0)}{2}=:z_{\widetilde{\sigma}}(t_0),
\]
and have we have that the considered option value equals to
\begin{align*}
\Pi\left(t_0,\Phi_{\text{call}}(V_{\widetilde{\sigma}}(T))\right)&=\mathbb{E}\left[e^{-r\,(T-t_0)}(V_{\widetilde{\sigma}}(T)-K)_+\Big|\mathcal{F}_{t_0}\right]\\
&=e^{-r\,(T-t_0)}\int_{z_{\widetilde{\sigma}}}^{+\infty}\Bigl\{v\,\exp\Bigl(r\,(T-t_0)-\varsigma(t_0)/2+x\Bigr)-K\Bigr\}\,f_{N(0,\varsigma(t_0))}(x)\,\mathrm{d}x\\
&=e^{-r\,(T-t_0)}\,v\,e^{r\,(T-t_0)-\varsigma(t_0)/2+\varsigma(t_0)/2}\left(1-N\left(\frac{z_{\widetilde{\sigma}}(t_0)-\varsigma(t_0)}{\sqrt{\varsigma(t_0)}}\right)\right)\\
&\qquad-K\,e^{-r\,(T-t_0)}\left(1-N\left(\frac{z_{\widetilde{\sigma}}(t_0)}{\sqrt{\varsigma(t_0)}}\right)\right)\\
&=v\,N\left(\frac{-z_{\widetilde{\sigma}}(t_0)+\varsigma(t_0)}{\sqrt{\varsigma(t_0)}}\right)-K\,e^{-r\,(T-t_0)}\,N\left(-\frac{z_{\widetilde{\sigma}}(t_0)}{\sqrt{\varsigma(t_0)}}\right).
\end{align*}
\qed\\
\end{proof}
\begin{remark}
Notice that the price of this call option depends on the future volatility, but, since it is deterministic, it is not an issue, and indeed we have obtained exact formulas.
\end{remark}

\begin{corollary}
Assuming that the risky asset dynamic follows a  geometric Brownian motion with time-dependent drift and volatility, see equation \eqref{risky-asset-t}, the price at time $t_0$ of a put option with payoff \eqref{payoff:p}, denoted as $\Phi_{\text{put}}$, linked to the MLVTS portfolio $V_{\widetilde{\sigma}}(t)$, is given by the following explicit formula
\begin{equation}\label{put:VTl}
\Pi\left(t_0,\Phi_{\text{put}}(V_{\widetilde{\sigma}}(T))\right)=K\,e^{-r\,(T-t_0)}\,N\left(-\widetilde{d}_2(t_0)\right)-v\,N\left(-\widetilde{d}_1(t_0)\right)
\end{equation}
where $\widetilde{d}_1$ and $\widetilde{d}_2$ are defined as in Proposition \ref{Prop:C:VTl}, and the proportion of portfolio value invested in the risky asset is $\widetilde{\alpha}(t):=\min\{L;\widehat{\sigma}/\sigma(t)\}$.
\end{corollary}

\begin{proof}
Direct consequence of the put-call parity formula.\qed
\end{proof}

\section{Greeks}\label{greeks}

In this section we will move on the quantitative study of the prices of options on a VTS portfolio in continuous time. In particular, we will explicitly derive the associated {\it Greeks'} values, the latter being those 
quantities representing derivatives price sensitivness to their characterizing underlying parameters' changes in time.

\bigskip
In what follows we are going to consider a risky asset evolving  as in the Black-Scholes model. Therefore, the price formulas \eqref{call:VT}, \eqref{put:VT}, \eqref{call:VTl} and \eqref{put:VTl} for call and put options, with VTS and MLVTS underlying portfolios, reduce to
\begin{align*}
\Pi\left(t_0,\Phi_{\text{call}}(V_{\widehat{\sigma}}(T))\right)=&v\,N(d_1)-K\,e^{-r\,(T-t_0)}\,N(d_2),\\
\Pi\left(t_0,\Phi_{\text{call}}(V_{\widetilde{\sigma}}(T))\right)=&
\begin{cases}
v\,N(\widetilde{d}_1)-K\,e^{-r\,(T-t_0)}\,N(\widetilde{d}_2),\quad&\text{for }\sigma<\widehat{\sigma}/L,\\
v\,N(d_1)-K\,e^{-r\,(T-t_0)}\,N(d_2),&\text{for }\sigma>\widehat{\sigma}/L,
\end{cases}\\
\Pi\left(t_0,\Phi_{\text{put}}(V_{\widehat{\sigma}}(T))\right)=&K\,e^{-r\,(T-t_0)}\,N(-d_2)-v\,N(-d_1),\\
\Pi\left(t_0,\Phi_{\text{put}}(V_{\widetilde{\sigma}}(T))\right)=&
\begin{cases}
K\,e^{-r\,(T-t_0)}\,N(-\widetilde{d}_2)-v\,N(-\widetilde{d}_1),\quad&\text{for }\sigma<\widehat{\sigma}/L,\\
K\,e^{-r\,(T-t_0)}\,N(-d_2)-v\,N(-d_1),&\text{for }\sigma>\widehat{\sigma}/L,
\end{cases}
\end{align*}
where
\begin{align*}
d_1=&\frac{1}{\widehat{\sigma}\,\sqrt{T-t_0}}\,\left(\log(v/K)+(r+\widehat{\sigma}^2/2)\,(T-t_0)\right),\\
d_2=&\frac{1}{\widehat{\sigma}\,\sqrt{T-t_0}}\,\left(\log(v/K)+(r-\widehat{\sigma}^2/2)\,(T-t_0)\right),\\
\widetilde{d}_1=&\frac{1}{\sigma\,L\,\sqrt{T-t_0}}\,\left(\log(v/K)+(r+L^2\,\sigma^2/2)\,(T-t_0)\right),\\
\widetilde{d}_2=&\frac{1}{\sigma\,L\,\sqrt{T-t_0}}\,\left(\log(v/K)+(r-L^2\,\sigma^2/2)\,(T-t_0)\right).
\end{align*}
\subsection{Vega}\label{S:Vega}
Since VTS portfolios are meant to preserve a fixed volatility level, the most representative Greek value is the {\it Vega} one, since it represents the sensitivity of the option price to the  risky asset's volatility.
\begin{proposition}\label{Vega:C}
The Vega of a call and put option, with payoff \eqref{payoff} and \eqref{payoff:p}, on VTS and MLVTS portfolios with weight strategies $\widehat{\alpha}=\widehat{\sigma}/\sigma$ and $\widetilde{\alpha}:=\min\{L;\widehat{\sigma}/\sigma\}$ respectively, are respectively given by
\begin{equation*}
\nu_{\{\Phi_{\text{call}},V_{\widehat{\sigma}}\}}=\partial_\sigma \Pi\left(t_0,\Phi_{\text{call}}(V_{\widehat{\sigma}}(T))\right)=0,
\end{equation*}
\begin{equation}\label{vega:VTl}
\nu_{\{\Phi_{\text{call}},V_{\widetilde{\sigma}}\}}=\partial_\sigma \Pi\left(t_0,\Phi_{\text{call}}(V_{\widetilde{\sigma}}(T))\right)=
\begin{cases}
\frac{v}{\sqrt{2\pi}}\,\exp\left(-\frac{\widetilde{d}_1^2}{2}\right)L\,\sqrt{T-t_0},\quad &\text{for }\sigma<\frac{\widehat{\sigma}}{L},\\
0,&\text{for }\sigma>\frac{\widehat{\sigma}}{L},
\end{cases}
\end{equation}
with $\nu_{\{\Phi_{\text{put}},V_{\widehat{\sigma}}\}}=\nu_{\{\Phi_{\text{call}},V_{\widehat{\sigma}}\}}$ and $\nu_{\{\Phi_{\text{put}},V_{\widetilde{\sigma}}\}}=\nu_{\{\Phi_{\text{call}},V_{\widetilde{\sigma}}\}}$, where 
\begin{equation*}
\widetilde{d}_1=\frac{\log(v/K)+\left(r+\frac{L^2\,\sigma^2}{2}\right)(T-t_0)}{L\,\sigma\,\sqrt{T-t_0}}\,.
\end{equation*}
\end{proposition} 

\begin{proof}
Let us consider the MLVTS for the case in which $\sigma<\frac{\widehat{\sigma}}{L}$, then we recall that the price of the call option simplifies to
\[
\Pi\left(t_0,\Phi_{\text{call}}(V_{\widehat{\sigma}}(T))\right)=v\,N(\widetilde{d}_1)-K\,\exp(-r\,(T-t))\,N(\widetilde{d}_2)\,,
\]
where 
\begin{align*}
\widetilde{d}_1&=\frac{-\log(K/v)+(r+L^2\,\sigma^2/2)\,(T-t)}{L\,\sigma\,\sqrt{T-t_0}},\\
\widetilde{d}_2&=\frac{-\log(K/v)+(r-L^2\,\sigma^2/2)\,(T-t)}{L\,\sigma\,\sqrt{T-t_0}}.
\end{align*}
Then, computing the partial derivative with respect to $\sigma$, we have
\begin{align*}
\partial_\sigma \Pi\left(t_0,\Phi_{\text{call}}(V_{\widetilde{\sigma}}(T))\right)&=
\frac{1}{\sqrt{2\pi}}\left(v\,e^{-\widetilde{d}_1^2/2}\left(L\,\sqrt{T-t_0}-\frac{\widetilde{d}_1}{\sigma}\right)+K\,e^{-\widetilde{d}_2^2/2-r\,(T-t_0)}\left(L\,\sqrt{T-t_0}+\frac{\widetilde{d}_2}{\sigma}\right)\right)\\
&=\frac{1}{\sqrt{2\pi}}\,v\,e^{-\widetilde{d}_1^2/2}\left(2\,L\,\sqrt{T-t_0}-\frac{\widetilde{d}_1-\widetilde{d}_2}{\sigma}\right)\\
&=\frac{v}{\sqrt{2\pi}}\,\exp\left(-\frac{(\widetilde{d}_1)^2}{2}\right)L\,\sqrt{T-t_0},
\end{align*}
where the second row steams from the identity $\exp(\frac{\widetilde{d}_2^2-\widetilde{d}_1^2}{2})=\frac{v}{K}\exp(r\,(T-t_0))$, and the last by $\widetilde{d}_1-\widetilde{d}_2=L\,\sigma\,\sqrt{T-t_0}$.

Similar computations also work for the Vega of put options.
\qed
\end{proof}

\begin{remark}
let us underline that if  $\sigma<\widehat{\sigma}/L$, then the MLVTS call and the standard leverage call share the same price, which can be expressed in terms of the standard call option price, denoted as $\Pi_S:=\Pi\left(t_0,\Phi_{\text{call}}(S(T))\right)$; hnce stating the dependence with respect the risky asset's volatility
\[
\Pi_{V_{\widetilde{\sigma}}}(\sigma)=\Pi_S(L\,\sigma),
\]
so that, computing the partial derivative with respect to the volatility, we have
\[
\nu_{\{\Phi_{\text{call}},V_{\widetilde{\sigma}}\}}=\partial_\sigma\Pi_{V_{\widetilde{\sigma}}}(\sigma)=L\,\partial_\sigma\Pi_{S}(L\,\sigma)\,,
\]
namely the same expression as in equation \eqref{vega:VTl}.
\end{remark}

In figure \ref{fig1} we provide a comparison between  the graphs of a Vega for a call option written on a portfolio adopting a MLVTS, and the Vega for a standard call option. The left graph represents Vegas  for at-the-money options. Here, one can notice that, while for volatilities higher than $\sigma>\frac{\widehat{\sigma}}{L}$, the MLVTS Vega is null, for {\it small} volatilities the MLVTS Vega is even higher than the Vega for standard call options. The right graph in figure \ref{fig1} represents the comparison taking into account the sensitivity of Vega with respect to the portfolio value. Notice that, while for  a standard call option the highest Vega is reached for the underlying share's value equal to
\[
v^*=s^*=K\,e^{-(T-t_0)\,(r-\frac{\sigma^2}{2})},
\] 
for the MLVTS call option, it is reached in
\[
v^*=K\,e^{-(T-t_0)\,(r-\frac{L^2\,\sigma^2}{2})}\,.
\]
Finally, in figure \ref{fig3}, we summarized the dependence of Vega with respect to both the volatility and the underlying portfolio value.


\begin{figure}[htbp]
\centering
\includegraphics[width=.48\textwidth]{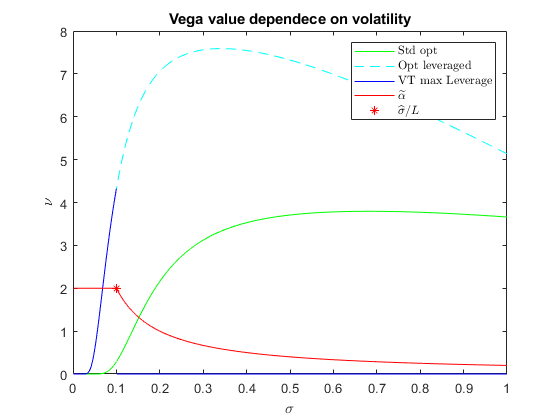}\,
\includegraphics[width=.48\textwidth]{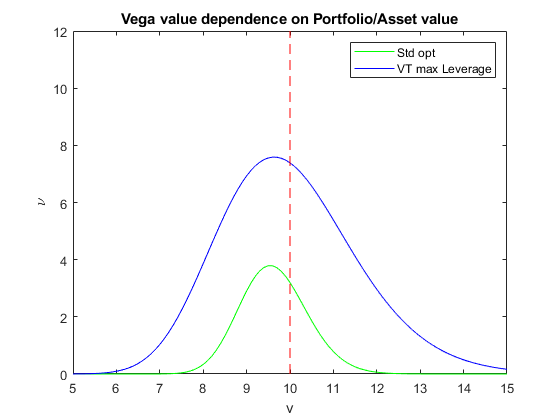}
\caption{The plots represent the behavior of  Vega values of an option written on the MLVTS portfolio, with maximum allowed leverage given by the weight strategy $\widetilde{\alpha}=\min(L,\widehat{\sigma}/\sigma)$, highlighting its dependence on the values of the volatility $\sigma$ (left) and on the values of $v$ (right). The parameters are set up as $r=5\%,v=12,K=10,t_0=0,T=1,\widehat{\sigma}=20\%,L=2$. 
For the volatility dependence (left), the MLVTS Vega line (in blue) is also compared with the dotted line of an hypothetical portfolio holding $L$ times its wealth in the risky asset (in cyan) and the Vega for a standard call option whose underlying is simply the risky asset. 
Instead, for the portfolio initial value dependence (right), we considered $\sigma=0.08$, i.e. $\sigma<\widehat{\sigma}/L$. The MLVTS Vega line (in blue) is the same as the Vega of an hypothetical portfolio holding $L$ times its wealth in the risky asset. This line is compared with the Vega of a portfolio investing all its capital in the risky asset. 
We remark  that we obtained a greater  Vega value for the MLVTS option  than the one for the standard option. This is due to the fact that we considered a relatively small volatility, namely less than $\widehat{\sigma}/L$). Instead, if we would have considered a volatility greater than $\widehat{\sigma}/L$, the Vega value for the option written on the MLVTS would have been identically zero.
}\label{fig1}
\end{figure}
\begin{figure}[htbp]
\centering
\includegraphics[width=.48\textwidth]{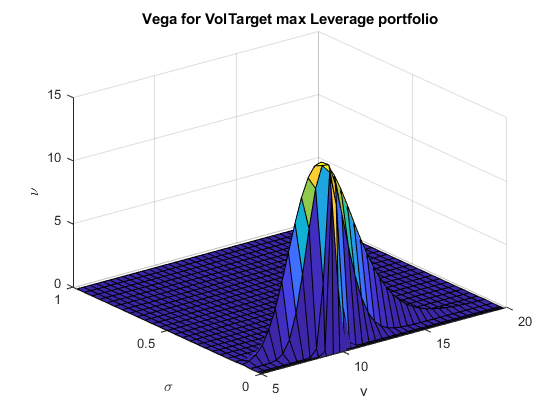}
\includegraphics[width=.48\textwidth]{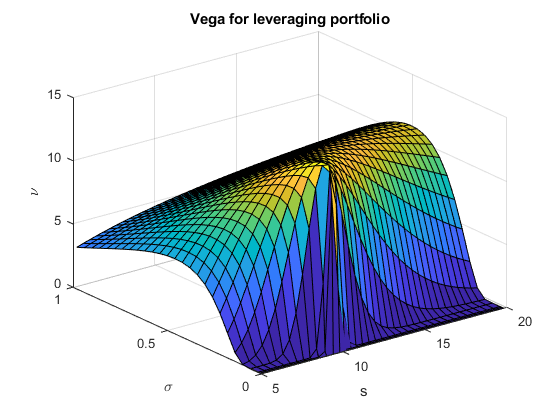}
\caption{The two surfaces are the Vega for a portfolio adopting a MLVTS (left figure) and the Vega for a portfolio investing $L$ times its wealth $v$ in the risky asset (right figure). One can notice that the MLVTS hedges well the portfolio against volatility variations when the volatility is high (higher than $\widetilde{\sigma}/L$).
The parameters are set as $r=5\%, t_0=0, T=1, \widehat{\sigma}=20\%, L=2, K=10$.
}\label{fig3}
\end{figure}

\subsection{Delta}\label{S:Delta}
Before dealing with the Delta of options written on VTS and MLVTS portfolios, it is worth to start analyzing the sensitivity of the VTS portfolio with respect to small changes in the risky asset price. To perform this, we write the VTS portfolio dynamics as
\begin{equation}\label{sf:VTp1}
\mathrm{d}V_{\widehat{\sigma}}(t)=\varphi_{\widehat{\sigma}}(t)\,\mathrm{d}S(t)+\psi_{\widehat{\sigma}}(t)\,\mathrm{d}B(t),
\end{equation}
where we defined $\varphi_S$ and $\varphi_B$ as the instantaneous number of shares and bonds held in the portfolio.
By the self-financing equation \eqref{sf:VTp}, we have that
\begin{align*}
\varphi_{\widehat{\sigma}}(t)&=\frac{V(t)\,\widehat{\alpha}(t)}{S(t)},\\
\psi_{\widehat{\sigma}}(t)&=\frac{V(t)\,(1-\widehat{\alpha}(t))}{B(t)},
\end{align*}
which means that the Delta of the VT portfolio is
\begin{equation}\label{Delta:V}
\Delta_{V_{\widehat{\sigma}}}=\frac{V(t)\,\widehat{\sigma}}{S(t)\,\sigma}.
\end{equation}

\begin{remark}\label{Rem:Ito}
Let us underline that the price of an option written on the VTS portfolio may be equivalently determined solely by the dynamics of  the VTS portfolio and the actual time, or by the risky asset dynamics, the bond dynamics and the actual time. Namely, we may denote the price of a generic option on a VTS portfolio as $\Pi(t,V)$ or $\Pi(t,S,B)$. 
Therefore, within the first setting, by equation \eqref{sf:VTp1} and It\^o-D\"oeblin formula, see \cite[Ch. 4]{Shreve}, we have
\begin{align}
\mathrm{d}\Pi(t,V)=&\partial_t\Pi(t,V)\,\mathrm{d}t+\partial_V\Pi(t,V)\,\mathrm{d}V_t+\frac{1}{2}\partial_{VV}^2\Pi(t,V)\,\mathrm{d}[V,V]_t\nonumber\\
=&\partial_t\Pi(t,V)\,\mathrm{d}t+\partial_V\Pi(t,V)\,\bigl(\varphi(t)\,\mathrm{d}S_t+\Psi(t)\,\mathrm{d}B_t\bigr)\nonumber\\
&\quad\quad\quad\quad\quad\quad\quad\quad\quad\quad\quad\quad+\frac{1}{2}\partial_{VV}^2\Pi(t,V)\,\varphi(t)^2\,\mathrm{d}[S,S]_t,\label{rem4.5:1}
\end{align}
where $[V,V]$ denotes the quadratic variation of the stochastic process $V$, see, e.g., \cite[3.4.2]{Shreve}.

Instead, considering the option price as a function of time, risky asset price and bond price, by It\^o-D\"oeblin formula, we have
\begin{equation}\label{rem4.5:2}
\mathrm{d}\Pi(t,S,B)=\partial_t\Pi(t,S,B)\,\mathrm{d}t+\partial_S\Pi(t,S,B)\,\mathrm{d}S_t+\partial_B\Pi(t,S,B)\,\mathrm{d}B_t+\frac{1}{2}\partial_{SS}^2\Pi(t,S,B)\,\mathrm{d}[S,S]_t\,,
\end{equation}
therefore, combining equations \eqref{rem4.5:1} and \eqref{rem4.5:2}, we derive a simpler expression for Delta and Gamma of an option on VTS portfolios
\begin{align*}
\partial_S\Pi(t,S,B)&=\partial_V\Pi(t,V)\,\varphi(t),\\
\partial_{SS}^2\Pi(t,S,B)&=\partial_{VV}^2\Pi(t,V)\,\varphi(t)^2.
\end{align*}
\end{remark}

\begin{proposition}\label{Delta:C}
The Delta of a European call option with payoff \eqref{payoff} on VTS and MLVTS portfolios with weight strategies $\widehat{\alpha}=\widehat{\sigma}/\sigma$ and $\widetilde{\alpha}:=\min\{L;\widehat{\sigma}/\sigma\}$ are respectively given by
\begin{equation}\label{delta:VT}
\Delta_{\{\Phi_{\text{call}},V_{\widehat{\sigma}}\}}=\partial_S \Pi\left(t_0,\Phi_{\text{call}}(V_{\widehat{\sigma}}(T))\right)=\frac{v\,\widehat{\sigma}}{s\,\sigma}\,N(d_1),
\end{equation}
\begin{equation}\label{delta:VTl}
\Delta_{\{\Phi_{\text{call}},V_{\widetilde{\sigma}}\}}=\partial_S \Pi\left(t_0,\Phi_{\text{call}}(V_{\widetilde{\sigma}}(T))\right)=
\begin{cases}
\frac{L\,v}{s}\,N(\widehat{d}_1),\quad &\text{for }\sigma<\frac{\widehat{\sigma}}{L},\\
\frac{v\,\widehat{\sigma}}{s\,\sigma}\,N(d_1),&\text{for }\sigma>\frac{\widehat{\sigma}}{L},
\end{cases}
\end{equation}
where 
\begin{align*}
d_1=\frac{\log(v/K)+\left(r+\frac{\widehat{\sigma}^2}{2}\right)(T-t_0)}{\widehat{\sigma}\,\sqrt{T-t_0}},\\
\widetilde{d}_1=\frac{\log(v/K)+\left(r+\frac{L^2\,\sigma^2}{2}\right)(T-t_0)}{L\,\sigma\,\sqrt{T-t_0}}.
\end{align*}
While the Delta of a European put option with payoff \eqref{payoff:p} is
\begin{equation}\label{delta:VT:p}
\Delta_{\{\Phi_{\text{put}},V_{\widehat{\sigma}}\}}=\partial_S \Pi\left(t_0,\Phi_{\text{put}}(V_{\widehat{\sigma}}(T))\right)=\frac{v\,\widehat{\sigma}}{s\,\sigma}\,(N(d_1)-1),
\end{equation}
\begin{equation}\label{delta:VTl:p}
\Delta_{\{\Phi_{\text{put}},V_{\widetilde{\sigma}}\}}=\partial_S \Pi\left(t_0,\Phi_{\text{put}}(V_{\widetilde{\sigma}}(T))\right)=
\begin{cases}
\frac{L\,v}{s}\,(N(\widehat{d}_1)-1),\quad &\text{for }\sigma<\frac{\widehat{\sigma}}{L},\\
\frac{v\,\widehat{\sigma}}{s\,\sigma}\,(N(d_1)-1),&\text{for }\sigma>\frac{\widehat{\sigma}}{L}.
\end{cases}
\end{equation}
\end{proposition} 

\begin{proof}
By It\^o calculus' chain rule (see Remark \ref{Rem:Ito})
\begin{equation}\label{DeltaC:p2}
\Delta_{\{\Phi_{\text{call}},V_{\widehat{\sigma}}\}}=\partial_V \Pi\left(t_0,\Phi_{\text{call}}(V_{\widehat{\sigma}}(T))\right)\,\Delta_{V_{\widehat{\sigma}}}.
\end{equation}
Moreover, it is straightforward to obtain
\begin{align}
\partial_V \Pi\left(t_0,\Phi_{\text{call}}(V_{\widehat{\sigma}}(T))\right)&=N(d_1)+v\,N'(d_1)\,\partial_vd_1-K\,e^{-r,(T-t_0)}N'(d_2)\,\partial_vd_2\nonumber\\
&=N(d_1)+\frac{1}{\sqrt{2\pi}\,\sqrt{T-t_0}\,\widehat{\sigma}^2\,v^2}\left(v\,e^{-\frac{1}{2}d_1^2}-K\,e^{-r\,(T-t_0)}e^{-\frac{1}{2}d_2^2}\right)\nonumber\\
&=N(d_1)+\frac{e^{-\frac{1}{2}d_1^2}}{\sqrt{2\pi}\,\sqrt{T-t_0}\,\widehat{\sigma}^2\,v}\left(1-\frac{K}{v}\,e^{-r\,(T-t_0)}e^{-\frac{1}{2}(d_2^2-d_1^2)}\right)=N(d_1),\label{DeltaC:p1}
\end{align}
where the last equality holds since $e^{-\frac{1}{2}(d_2^2-d_1^2)}=v/K\,e^{r\,(T-t_0)}$. Therefore, substituting \eqref{Delta:V} and \eqref{DeltaC:p1} in equation \eqref{DeltaC:p2}, we obtain \eqref{delta:VT}.

The same arguments works for $\Delta_{V_{\widetilde{\sigma}}}$, just taking care that in this case we have 
\[
\Delta_{V_{\widetilde{\sigma}}}=\frac{L\,v}{s},\quad\text{for }\sigma<\frac{\widehat{\sigma}}{L},
\]
since for $\sigma<\frac{\widehat{\sigma}}{L}$, $\widetilde{\alpha}=L$.

For the put option case, we have that $\partial_V \Pi\left(t_0,\Phi_{\text{call}}(V_{\widehat{\sigma}}(T))\right)=N(d_1)-1$. Therefore, analogously as before, we obtain equation \eqref{delta:VT:p}.
\qed
\end{proof}

In figure \ref{delta} we compared the Deltas for calls and puts written on a VTS portfolio with calls and puts written on the risky asset. One can notice that Deltas for VTS-linked options present  an asymptotic behavior for  both low and high volatilities. This is because, for extreme low volatilites, the VTS portfolio finances a great amount of shares through short selling the riskless asset, while, for extreme high volatilities,
 the VTS portfolio invests only a small proportion of its value in the risky asset. See Remark \ref{Rem:Ito} for the mathematical explanation of this effect.

For MLVTS portfolios the analysis of the Deltas with respect to the volatility of the risky asset is not significantly different than the Deltas for standard options, since it suffices to consider the change of variable $\sigma_L=L\,\sigma$ in the standard Delta value and to boost the latter by the maximum leverage parameter $L$.

\begin{figure}[htbp]
\centering
\includegraphics[width=0.48\textwidth]{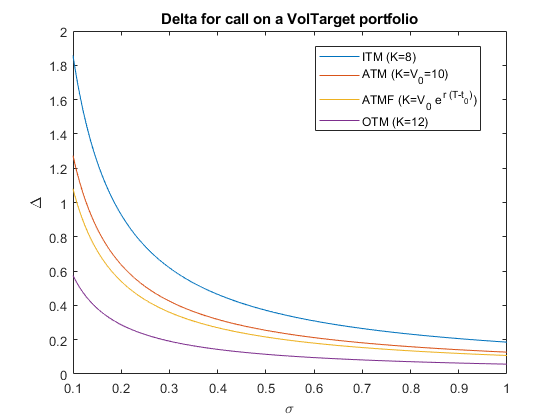}\includegraphics[width=0.48\textwidth]{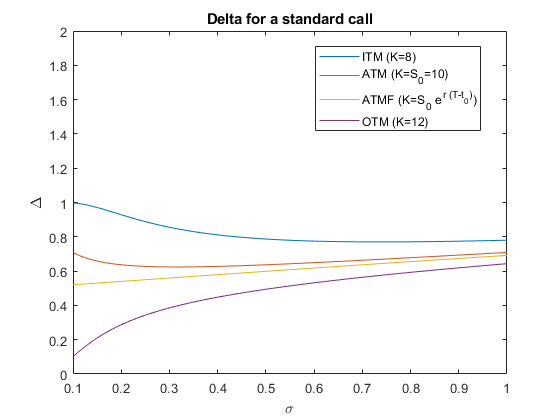}
\includegraphics[width=0.48\textwidth]{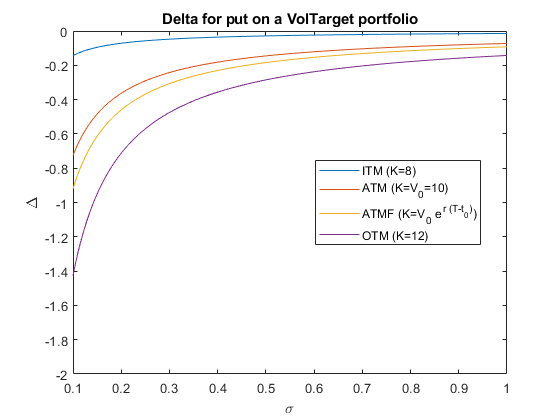}\includegraphics[width=0.48\textwidth]{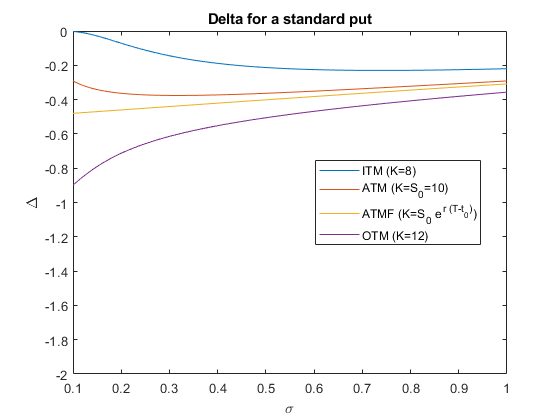}
\caption{The right graphs represent the behavior of the Delta of standard call (top figure) and put (bottom figure) options with respect to different volatility values. The graphs on the left represent the Delta for call and put options written on VTS portfolios. We considered strike prices $K$ in order to obtain an {\it in-the-money} option, an {\it at-the-money} option, an {\it at-the-money-forward} option and an {\it out-of-the-money} option. Notice that the Delta for the VTS-linked options exhibits two asymptotes: the vertical one corresponding to null volatility, and the horizontal one which correspond to a volatility value that goes to infinity. The parameters are fixed as $s=v=10$, $\widehat{\sigma}=0.2$, $\mu=8\%$, $r=5\%$, $T=1$, $t_0=0$ and the volatilities start at $\sigma=0.1$.}\label{delta}
\end{figure}

\subsection{Gamma}
The computation of the Gamma for options on VTS and MLVTS portfolios can be derived as 
in section \ref{S:Delta}, see, in particular, Remark \ref{Rem:Ito}.

\begin{proposition}\label{Gamma:C}
The Gamma of an option with payoff \eqref{payoff} on the VTS and MLVTS portfolios with weight strategies $\widehat{\alpha}=\widehat{\sigma}/\sigma$ and $\widetilde{\alpha}:=\min\{L;\widehat{\sigma}/\sigma\}$ are respectively given by
\begin{equation}\label{gamma:VT}
\Gamma_{\{\Phi_{\text{call}},V_{\widehat{\sigma}}\}}=\partial_{SS}^2 \Pi\left(t_0,\Phi_{\text{call}}(V_{\widehat{\sigma}}(T))\right)=\frac{v\,\widehat{\sigma}}{s^2\,\sigma^2\,\sqrt{T-t_0}}\,f_{N(0,1)}(d_1),
\end{equation}
\begin{equation}\label{gamma:VTl}
\Gamma_{\{\Phi_{\text{call}},V_{\widetilde{\sigma}}\}}=\partial_{SS}^2 \Pi\left(t_0,\Phi_{\text{call}}(V_{\widetilde{\sigma}}(T))\right)=
\begin{cases}
\frac{L\,v}{s^2\,\sigma\,\sqrt{T-t_0}}\,f_{N(0,1)}(\widehat{d}_1),\quad &\text{for }\sigma<\frac{\widehat{\sigma}}{L},\\
\frac{v\,\widehat{\sigma}}{s^2\,\sigma^2\,\sqrt{T-t_0}}\,f_{N(0,1)}(d_1),&\text{for }\sigma>\frac{\widehat{\sigma}}{L},
\end{cases}
\end{equation}
with $\Gamma_{\{\Phi_{\text{put}},V_{\widehat{\sigma}}\}}=\Gamma_{\{\Phi_{\text{call}},V_{\widehat{\sigma}}\}}$ and $\Gamma_{\{\Phi_{\text{put}},V_{\widetilde{\sigma}}\}}=\Gamma_{\{\Phi_{\text{call}},V_{\widetilde{\sigma}}\}}$, where by $f_{N(0,1)}$ we denote the probability density function of a standard normal random variable and 
\begin{align*}
d_1=\frac{\log(v/K)+\left(r+\frac{\widehat{\sigma}^2}{2}\right)(T-t_0)}{\widehat{\sigma}\,\sqrt{T-t_0}},\\
\widetilde{d}_1=\frac{\log(v/K)+\left(r+\frac{L^2\,\sigma^2}{2}\right)(T-t_0)}{L\,\sigma\,\sqrt{T-t_0}}.
\end{align*}
\end{proposition} 

\begin{proof}
By Remark \ref{Rem:Ito}, we have  $\partial_{SS}^2\Pi(t,S,B)=\partial_{VV}^2\Pi(t,V)\,\varphi(t)^2$. Computing $\partial_{VV}^2\Pi(t,V)$:
\begin{align*}
\partial_{VV}^2\Pi(t,V)&=\partial_V\left[\partial_V \Pi\left(t_0,\Phi_{\text{call}}(V_{\widehat{\sigma}}(T))\right)\right]\\
&=\partial_V N(d_1)\\
&=N'(d_1)\,\partial_Vd_1\\
&=f_{N(0,1)}(d_1)\,\frac{1}{\widehat{\sigma}\,v\,\sqrt{T-t_0}}\,,
\end{align*}
and since $\varphi(t_0)=\frac{v\,\widehat{\sigma}}{s\,\sigma}$, we obtain \eqref{gamma:VT}.

The Gamma for put VTS-linked options is the same as the Gamma for call VTS-linked options, since the second partial derivatives w.r.t.\,the portfolio value of the price of the two options are the same.
\qed
\end{proof}

In figure \ref{gamma}, we compared the Gamma for standard European options with the Gamma for European options written on VTS portfolios. Notice that, while the Gamma for standard European options exhibits two asymptotes only when the underlying risky asset is {\it at-the-money-forward} (ATMF), i.e. $S=K\,e^{-r\,(T-t_0)}$ which implies
\[
\Gamma_S=\frac{1}{s\,\sigma\,\sqrt{T-t_0}}e^{-\frac{1}{2}\frac{\sigma^2\,(T-t_0)}{4}}\\,,
\]
for the Gamma of VTS-linked options we have always two asymptotes, since for low volatilities also their Gamma is amplified, in fact even more than the Delta. 

\begin{figure}[htbp]
\centering
\includegraphics[width=0.48\textwidth]{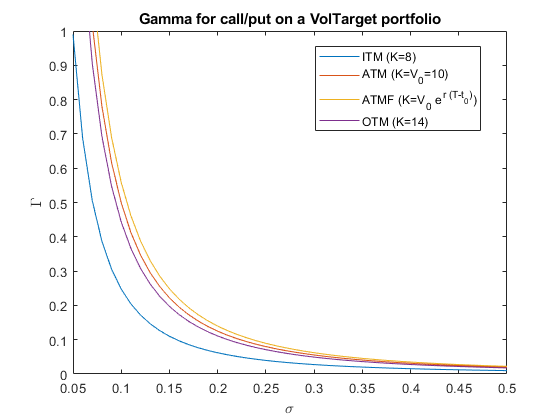}\includegraphics[width=0.48\textwidth]{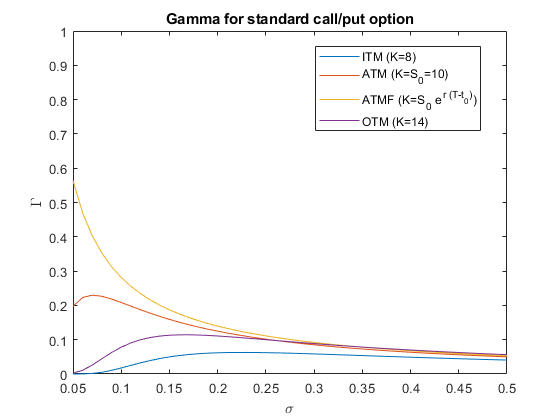}
\caption{The graph on the right represents the behavior of the Gamma of standard call/put options with respect to different volatility values, while the one on the left represents the Gamma for call/put options written on VTS portfolios. As in figure \ref{delta}, we considered strike prices $K$ in order to obtain an {\it in-the-money} option, an {\it at-the-money} option, an {\it at-the-money-forward} option and an {\it out-of-the-money} option. Notice that, once again, also this Greek for the VT options exhibits asymptotes, while this is the case for standard options only when the underlying asset is ATMF. The parameters are fixed as $v=s=10$, $\widehat{\sigma}=0.2$, $\mu=8\%$, $r=5\%$, $T=1$, $t_0=0$ and the volatilities start at $\sigma=0.05$.}\label{gamma}
\end{figure}

%

\section{Numerical simulations}\label{numerics}
To better explain how VTS and MLVTS portfolios work, let us assume that the dynamics of the risky asset evolve according to the Heston model, see, e.g., \cite{Heston} and \cite{Barbu}, for further details and control-theory related problems.
The, we have
\begin{align}
\mathrm{d}S_t&=\mu\,S_t\,\mathrm{d}t+\sqrt{\nu_t}\,S_t\,\mathrm{d}W_t^{(1)},\label{H1}\\
\mathrm{d}\nu_t&=\kappa\,(\theta-\nu_t)\,\mathrm{d}t+\xi\,\sqrt{\nu_t}\,\mathrm{d}W_t^{(2)}\label{H2},
\end{align}
where $W^{(1)}$ and $W^{(2)}$ are $\rho$-correlated Brownian motions, $\nu$ evolves as a Cox-Ingersoll-Ross (CIR) process  representing the instantaneous variance of the risky asset, $\theta$ is the long-variance, $\kappa$ is the rate at which $\nu$ reverts to $\theta$, $\xi$ is the volatility of the volatility, and we assume that the Feller condition holds:
\[
2\,\kappa\,\theta>\xi^2\,,
\]
in order to guarantee the process $\nu$ to be strictly positive. 

\bigskip

Let us consider underlying risky asset's parameters calibrated to values observed in the real data, as in the papers by Morellec et al. \cite{Morellec}, and CIR parameters as in the seminal paper by Samuelson \cite{Samuelson}. Consider the parameters values as shown in Table \ref{tab:par}, namely an adaptation of the ones in \cite{Morellec,Samuelson}, to show representative scenarios explaining the effect of the VTS and the MLVTS.

\begin{table}[h]
\centering
\begin{tabular}{|c|ccccccccccc|}
\hline
Figure &$\kappa$ & $\theta$ & $\xi$ & $\rho$ & $\nu_0$ & $\mu$ & $S_0$ & $r$ & $B_0$ & $V_0$ & $T$\\
\hline
\ref{Heston:VT} & 0.6067 & 0.2207 & 0.2928 & -0.75 & 0.2154 & 8.24\% & 100 & 4.2\% & 20 & 100 & 1\\
\ref{Heston:VT-VTL} & '' & 0.1707 & '' & '' & 0.1654 & '' & '' & '' & '' & '' & ''\\
\hline
\end{tabular} 
\caption{The parameters of the Heston model that we considered in order to simulate the asset dynamics.}\label{tab:par}
\end{table}

We partition the time-interval $[0,T]$ into $N$ equal subintervals of width $T/N$
\[
0=t_0<t_1<\dots<t_N=T
\]
and simulate a realization of the bivariate process $(S,\nu)_t$.
By a modified {\it Euler-Maruyama scheme}, see below, we approximate the path of the corresponding VTS and MLVTS portfolio
\begin{equation}\label{EM:scheme}
V_{\widehat{\sigma}}^{\Delta t}(t_{n+1})=V_{\widehat{\sigma}}^{\Delta t}(t_{n})\left\{1+\frac{\alpha(t_n)}{S(t_n)}\Delta S_n+\frac{1-\alpha(t_n)}{B(t_n)}\,\Delta B_n\right\},\quad\text{for }n\in\{0,\dots,N-1\},
\end{equation}
where  $\Delta S_n:=S(t_{n+1})-S(t_n)$, $\Delta B_n:=B(t_{n+1})-B(t_n)$ and $V^{\Delta t}(0)=v$.

\begin{proposition}
Let $T>0$ be fixed. The numerical scheme \eqref{EM:scheme} is strongly convergent to the solution to \eqref{wealth}., i.e.
\[
\lim_{\Delta t\rightarrow0}\mathbb{E}[|V_T-V_T^{\Delta t}|]=0.
\]
\end{proposition}

\begin{proof}
The Euler-Maruyama scheme associated to equation \eqref{wealth} is 
\begin{equation}\label{proof:EM2}
V_{\widehat{\sigma}}^{\Delta t}(t_{n+1})=V_{\widehat{\sigma}}^{\Delta t}(t_{n})\left\{1+\left[\frac{\widehat{\sigma}}{\sqrt{\nu(t_n)}}\,(\mu-r)-r\right]\Delta t+\widehat{\sigma}\,\Delta W_n^{(1)}\right\},
\end{equation}
which i strongly convergent to \eqref{wealth}.
Moreover, we have that 
\begin{equation}\label{proof:EM1}
\lim_{\Delta t\rightarrow 0}\mathbb{E}\left[\frac{\frac{\Delta S_n}{S(t_n)}-\mu\,\Delta t}{\sqrt{\nu(t_n)}}-\Delta W_n^{(1)}\right]=0.
\end{equation}
Substituting \eqref{proof:EM1} in \eqref{proof:EM2}, we obtain \eqref{EM:scheme}.
\qed
\end{proof}

Even if the  Euler-Maruyama scheme is  strongly convergent, it is {\it just} of order 0.5, therefore, aiming at achivieng a higher convergence order,
we also considered the following modification of the {\it Milstein scheme}:
\begin{multline*}
V_{\widehat{\sigma}}^{\Delta t}(t_{n+1})=V_{\widehat{\sigma}}^{\Delta t}(t_n)\,\Biggl\{1+\alpha(t_n)\frac{\Delta S_n}{S(t_n)} +(1-\alpha(t_n))\frac{\Delta B_n}{B(t_n)}\\-\frac{\alpha(t_n)\,(1-\alpha(t_n))}{2}\left[\left(\frac{\Delta S_n}{S(t_n)}\right)^2-\nu(t_n)\,\Delta t\right]\Biggr\},
\end{multline*}
which in general converges strongly to the solution with order 1.

\subsection{Analysis}
Let us underline  that the sensitivity analysis provided along the previous section is well captured by figure \ref{Heston:VT-VTL}. For example, one can notice that the tendency of the path is more or less met, dependently on the volatility instantaneous value, see Section \ref{S:Delta} for the  sensitivity analysis of the VTS portfolio with respect to small changes in the underlying risky asset, i.e. the Delta analysis. Moreover, it is clearly visible that the white noise is affecting the VTS portfolio value linearly, i.e.\,the Vega for the VTS standard portfolio is null, see Section \ref{S:Vega}.

\begin{figure}[htbp]
\centering
\includegraphics[width=\textwidth]{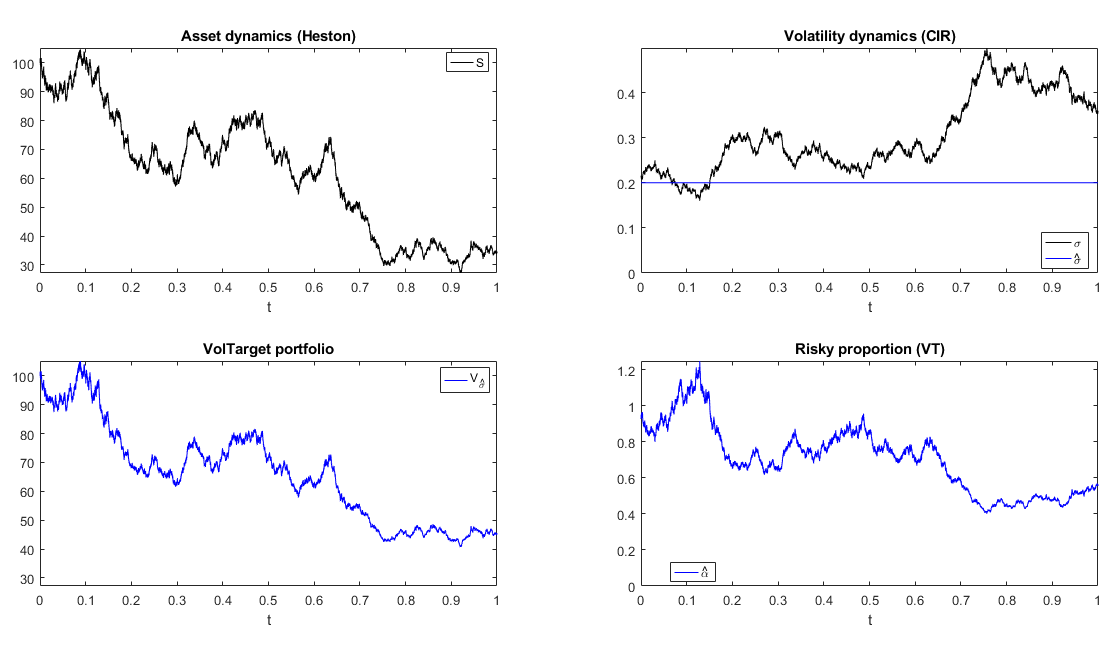}
\caption{The top graphs represent the asset price and the volatility values, simulated as a realization of a Heston model. Here the volatility (top-right figure) is more frequently greater than the target volatility $\widehat{\sigma}=0.2$, and in these cases the risky proportion $\widehat{\alpha}$ is less than one (bottom-right figure). In the bottom-left figure is represented the corresponding realization of the VTS portfolio. For the discretization scheme \eqref{EM:scheme} we considered a time step $\Delta t=10^{-6}$.}\label{Heston:VT}
\end{figure}


\begin{figure}[htbp]
\centering
\includegraphics[width=\textwidth]{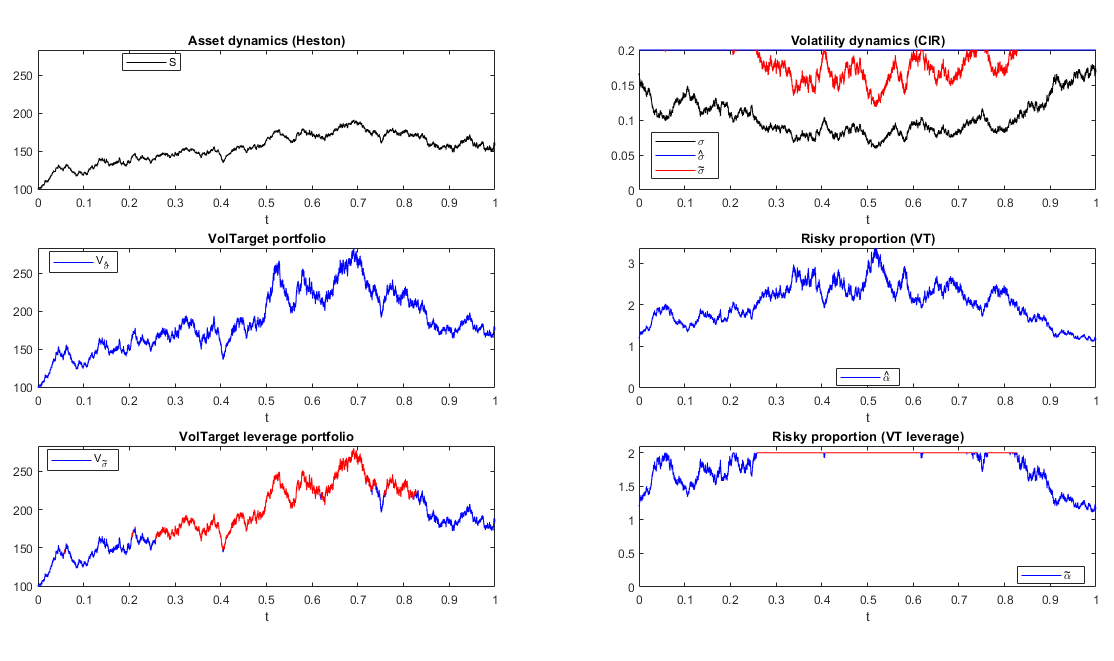}
\caption{The left graphs from the top to the bottom represent the risky asset dynamics, the VT portfolio and the MLVTS portfolio, for $\widehat{\sigma}=0.2$ and $L=2$. The top-right figure represents the volatility of the risky asset (in black), the volatility of the VTS portfolio (in blue) and the effect of the leverage limitation in the MLVTS (in red). In the bottom-left figure are highlighted in red the path section of the MLVTS portfolio in which the leverage effect intervenes. For the discretization scheme \eqref{EM:scheme} we considered a time step $\Delta t=10^{-6}$.}\label{Heston:VT-VTL}
\end{figure}

\section{Extension to the transaction case and concluding remarks}
The present paper presents a first attempt to consider options linked to VTSs from an analytical perspective. We develop closed-end formulas for call and put options linked to VolTarget concepts, as well as for the associated sensitiveness, {\it Greeks}, parameters. 

The results agree with what we would expect from a practitioner view. One can see, how a VolTarget approach can simplify option pricing for structured products and why also key hedging parameters look much easier than for standard options with changing volatility pattern. 
Further analysis should be done relaxing some of the assumptions we made to derive our results, e.g. dropping the {\it non transaction costs} hypotesis.


As an example, we already started to look into the aspect of transaction costs and how these can be embedded into our framework. We hereby point out two possible ways that can be undertaken. The first one consists in a modification of the chosen VTS, i.e. the VTS portfolio will no longer pursue a constant volatility, instead it will aim to have a volatility belonging to a desired interval. Moreover, one could consider to deal with a structural modification, namely considering a restriction of the admissible time interval  in which we will have portfolio weight adjustments, to consider a discrete subset. Namely, the VTS will pursue the target volatility only in a discrete set of time points, instead of considering continuous adjustments. Such modifications are required, when the asset dynamics are assumed to not have a constant volatility, in order to avoid the cumulated transaction costs to be theoretically infinity, which could happen even when the transaction costs are relatively small, see, e.g., 
\cite{Boyle1, Boyle2, Gilster, Leland, Mello}.

We will also investigate how a dynamic asset allocation strategy can be developed within real world scenarios, when a rather constant volatility level can be considered. 


\end{document}